\documentclass[twocolumn]{article}
\usepackage{preprint}

\usepackage[utf8]{inputenc} 
\usepackage[T1]{fontenc}    
\usepackage{cite}
\usepackage[]{graphicx}
\usepackage[x11names]{xcolor}
\usepackage{textcomp}

\usepackage{amssymb, amsthm, amsmath, latexsym, amsfonts, mathrsfs, amsbsy,mathtools,relsize}
\usepackage{listings}
\usepackage{enumerate}
\usepackage{soul}
\usepackage{pgfplots}
\usepackage{tikz}
\usepackage{float}
\usepackage{dsfont}
\usepackage{verbatim}

\usepackage[colorlinks]{hyperref}
\usepackage{balance}

\usepackage[ruled,vlined,linesnumbered]{algorithm2e}
\DontPrintSemicolon
\SetKwInput{KwIn}{Input}
\SetKwInput{KwOut}{Output}
\SetKwInput{KwInit}{Initialization}

\hypersetup{
    citecolor=Firebrick4,
    linkcolor=Firebrick4,   
    urlcolor=Red3}
\usetikzlibrary{shapes,shapes.geometric,arrows,arrows.meta,fit,calc,positioning,automata,through,intersections}
\tikzset{diamond state/.style={draw,diamond}}

\newtheoremstyle{theoremdd}
  {\topsep}
  {\topsep}
  {\itshape}
  {0pt}
  {\bfseries}
  {.}
  { }
  {\thmname{#1}\thmnumber{ #2}\textnormal{\thmnote{ (#3)}}}

\theoremstyle{theoremdd}
\newtheorem{theorem}{Theorem}
\newtheorem{lemma}{Lemma}
\newtheorem{assumption}{Assumption}

\newtheorem{remark}{Remark}

\newcommand{\bm}[1]{\boldsymbol{#1}} 
\newcommand{\set}[1]{\mathcal{#1}} 
\newcommand{\ie}{\textit{i.e.,~}} 
\newcommand{\eg}{\textit{e.g.,~}} 
\newcommand{\inneighbor}[1]{\set{N}_{#1}^{\texttt{in}}}
\newcommand{\outneighbor}[1]{\set{N}_{#1}^{\texttt{out}}}
\newcommand{\indegree}[1]{d_{#1}^{\texttt{in}}}
\newcommand{\outdegree}[1]{d_{#1}^{\texttt{out}}}

\newenvironment{listEM}{
\begin{list}{$\bullet$}{%
    \setlength{\itemsep}{0.5ex}
    \setlength{\labelsep}{0.5ex}
    \setlength{\labelwidth}{5ex}
    \setlength{\parsep}{0in} 
    \setlength{\parskip}{0in}
    \setlength{\topsep}{0in} 
    \setlength{\partopsep}{0in}
    \setlength{\leftmargin}{3.5ex}}}
{\end{list}}

\begin{document}

\title{Cooperative Bandit Learning in Directed Networks with Arm-Access Constraints}


\author{
  Evagoras Makridis\\
  Department of Electrical and Computer Engineering\\
  University of Cyprus\\
   Nicosia, Cyprus\\
  \texttt{makridis.evagoras@ucy.ac.cy} \\
   \And
  Themistoklis Charalambous\thanks{Themistoklis Charalambous is also a visiting professor at the Department of Electrical Engineering and Automation, School of Electrical Engineering, Aalto University, Espoo, Finland.}\\
  Department of Electrical and Computer Engineering\\
  University of Cyprus\\
   Nicosia, Cyprus\\
  \texttt{charalambous.themistoklis@ucy.ac.cy}\vspace{0.5cm}
}


\maketitle

\begin{abstract}
Sequential decision-making under uncertainty often involves multiple agents learning which actions (arms) yield the highest rewards through repeated interaction with a stochastic environment. This setting is commonly modeled by cooperative multi-agent multi-armed bandit problems, where agents explore and share information without centralized coordination. In many realistic systems, agents have heterogeneous capabilities that limit their access to subsets of arms and communicate over asymmetric networks represented by directed graphs. In this work, we study multi-agent multi-armed bandit problems with partial arm access, where agents explore and exploit only the arms available to them while exchanging information with neighbors. We propose a distributed consensus-based upper confidence bound (UCB) algorithm that accounts for both the arm accessibility structure and network asymmetry. Our approach employs a mass-preserving information mixing mechanism, ensuring that reward estimates remain unbiased across the network despite accessibility constraints and asymmetric information flow. Under standard stochastic assumptions, we establish logarithmic regret for every agent, with explicit dependence on network mixing properties and arm accessibility constraints. These results quantify how heterogeneous arm access and directed communication shape cooperative learning performance.
\end{abstract}

\keywords{multi-armed bandits, distributed, multi-agent, directed graphs, arm-access constraints}

\section{Introduction}\label{sec:introduction}

Multi-armed bandit (MAB) problems constitute a fundamental framework for sequential decision-making under uncertainty. In this setting, a learner (often referred to as \emph{agent}) repeatedly selects one option from a finite set of alternatives, called \emph{arms}, where each arm corresponds to an action associated with an unknown reward distribution. At every time step, the learner chooses an arm, observes a stochastic reward drawn from its distribution, and aims to maximize the total accumulated reward over time, or equivalently, to minimize cumulative regret with respect to the best arm in hindsight. A central challenge in MAB problems is the exploration–exploitation trade-off. This trade-off suggests that the learner must explore different arms to estimate their expected rewards while exploiting the currently best-performing arm to collect reward. For a nice introduction and review on MAB problems, the reader is recommended \cite{slivkins2019introduction} and \cite{sutton1998reinforcement}.

Classical formulations consider a single decision maker with full access to all arms \cite{lai1985asymptotically,auer2002finite}. However, many modern applications, including sensor networks, multi-robot systems, and collaborative recommendation systems, involve multiple agents that learn in a cooperative way while communicating over a network. In such cooperative settings, information is inherently distributed and must be aggregated through local interactions, often via consensus-based mechanisms \cite{landgren2016bdistributed,martinez2019decentralized}, while agents share the same reward distribution for each arm, hereinafter referred to as \emph{homogeneous reward setting}.

Early works on cooperative bandits typically assume that the communication network is undirected and connected, allowing agents to asymptotically agree on global statistics via symmetric averaging dynamics \cite{landgren2016adistributed,landgren2021distributed,zhu2021distributed}. Under these conditions, consensus mechanisms enable each agent to approximate centralized empirical means and pull counts, thereby achieving regret guarantees comparable to a single centralized learner.

\begin{figure}[t!]
\centering
\includegraphics[scale=0.85]{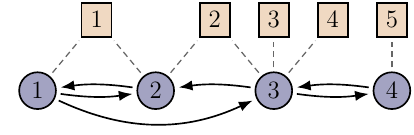}
\caption{Example of the cooperative multi-agent multi-armed bandit framework with arm-accessibility constraints. Blue circular nodes represent agents and beige rectangular nodes represent arms. Directed solid edges indicate communication links among agents, while dashed edges denote arm-accessibility relations between agents and arms.}
\label{fig:framework}
\end{figure}

In more general scenarios, however, communication networks are directed and possibly unbalanced, reflecting asymmetric information flow and heterogeneous influence among agents. The works in \cite{zhu2023distributed,zhu2024decentralized}, considered the cooperative MAB problem in directed graphs, where consensus dynamics no longer preserve averages, and the asymptotic network estimate depends on the stationary distribution of the underlying stochastic matrix. As a result, agents may contribute unequally to the global learning process, and additional compensation mechanisms are required to ensure that distributed learning algorithms remain stable and achieve sublinear regret.

While the aforementioned existing works address communication asymmetry, they assume that every agent has access to all arms. In many practical distributed systems, however, this assumption is restrictive. Agents may be geographically separated, equipped with heterogeneous sensors, or subject to resource constraints, so that each agent can sample only a subset of the global arm set, as shown in the example of Fig.~\ref{fig:framework}. Introducing arm-accessibility constraints fundamentally alters the cooperative learning problem. First, no single agent may have direct access to the globally optimal arm. Second, the rate at which an arm can be learned is influenced by the diversity of agents’ observations, \ie the more agents contributing independent samples, the faster the empirical mean converges. In directed and unbalanced graphs, this interaction becomes particularly apparent. For example, if an arm is observable only by agents with low network influence, information about that arm propagates slowly through the network, creating structural learning bottlenecks.

This paper studies cooperative multi-agent multi-armed bandits where each agent has access only to a subset of arms and must rely on information sharing to learn about the others. Unlike classical cooperative bandits in which every agent can observe all arms, the learning capability here is fundamentally shaped by which agents can access which arms. This coupling introduces new structural bottlenecks that fundamentally affect identifiability and learning speed. In fact, agents in our proposed distributed framework with arm-accessibility constraints cooperate only to reduce learning regret, not to change feasibility constraints. The main contributions of this work are as follows:
\begin{listEM}
	\item We consider a cooperative bandit framework with heterogeneous arm accessibility across agents. This is captured by an agent–arm incidence matrix, which encodes the set of arms available to each agent. The resulting accessibility constraints introduce a structural performance limitation that is independent of the learning dynamics. We formally characterize the interplay between accessibility and network topology, and distinguish between intrinsic structural loss and the regret due to learning.
	\item We employ a running sum-preserving ratio consensus algorithm, which allows each agent to track the true network-wide empirical mean reward of every arm, even if it cannot directly pull that arm. This property distinguishes our approach from other cooperative bandit algorithms, where information mixing may bias the estimated arm means.
	\item We design a distributed UCB learning policy based on running ratio consensus and network-aware confidence bounds. Despite asymmetric communication and heterogeneous arm accessibility, we prove logarithmic regret for every agent with explicit dependence on structural accessibility and network mixing properties.
\end{listEM}


\section{Preliminaries}\label{sec:background}

\subsection{Mathematical Notation}
The sets of real, integer, and natural numbers are denoted as $\mathbb{R}$, $\mathbb{Z}$, and $\mathbb{N}$, respectively. The set of nonnegative real (integer) numbers is denoted as $\mathbb{R}_{\geq 0}$ ($\mathbb{Z}_{\geq0}$). The nonnegative orthant of the $n$-dimensional real space $\mathbb{R}^n$ is denoted as $\mathbb{R}_{\geq 0}^{n}$. 
Matrices are denoted by capital letters, and vectors by small letters. The transpose of a matrix $A$ and a vector $x$ are denoted as $A^\top$, $x^\top$, respectively. The all-ones and all-zeros vectors are denoted by $\mathbf{1}$ and $\mathbf{0}$, respectively, with their dimensions being inferred from the context. 
Throughout the paper, $x^k_i(t)$ denotes a quantity associated with arm $k$ at agent $i$ and time $t$, unless otherwise stated. When a variable does not depend on a particular index, that index is omitted. For example, $x^k(t)$ denotes a network-level quantity for arm $k$ at time $t$, $x_i(t)$ denotes an agent-specific quantity at agent $i$ at time $t$, and $x^k$ denotes a time-invariant quantity associated with arm $k$.

\subsection{Network Model}
Consider a network composed of $N\ge2$ agents, for which we denote the set of agents by $\set{V}:=\{1, 2, \ldots, N\}$, with $N=\lvert \set{V}\rvert$. Each agent can distinguish the set of its neighboring agents with which it can interact by means of local communication. The agent interconnection topology can be represented by a digraph $\set{G}=(\set{V}, \set{E})$, where $\set{E} \subseteq \set{V} \times \set{V}$ describes the communication interaction between agents, \ie the information flow between the agents in $\set{V}$. More specifically, if $\varepsilon_{ij} := (j,i)\in\set{E}$, then agent $i$ can receive information from agent $j$, and thus agent $j$ can be considered as in-neighbor of agent $i$. Formally, agents that can directly transmit information to agent $j$ are called in-neighbors of agent $j$, and belong to the set 
$$\inneighbor{j}=\{i \in \set{V} \mid \varepsilon_{ji} \in \set{E}\}.$$
The number of agents in the in-neighborhood set is called in-degree and is denoted by $\indegree{j} = \left|\inneighbor{j}\right|$. The agents that can directly receive information from agent $j$ are called out-neighbors of agent $j$, and belong to the set $$\outneighbor{j}=\{l \in \set{V} \mid \varepsilon_{lj} \in \set{E}\}.$$
The number of agents in the out-neighborhood set is called out-degree and is denoted by $\outdegree{j}= \left|\outneighbor{j}\right|$. Here, it is natural to assume that each agent $j \in \set{V}$ has immediate access to its own information, and thus we assume that the corresponding self-loop is available $\varepsilon_{jj} \in \set{E}$, and it is not included in the agents' out-neighborhood and in-neighborhood sets. 

A directed path from $i$ to $l$ with a length of $t$ exists if a sequence of nodes $i \equiv l_0,l_1, \dots, l_t \equiv l$ can be found, satisfying $(l_{\tau+1},l_{\tau}) \in \mathcal{E}$ for $ \tau = 0, 1, \dots , t-1$. In $\mathcal{G}$ an agent $i$ is reachable from an agent $j$ if there exists a path from $j$ to $i$ which respects the direction of the edges. The digraph $\mathcal{G}$ is said to be strongly connected if every agent is reachable from every other agent.

\subsection{Problem Formulation}\label{sec: problem}

Consider a network of $N$ agents modeled by a directed graph $\mathcal{G} = (\mathcal{V}, \mathcal{E})$, where $\mathcal{V} = \{1,\dots,N\}$ denotes the set of agents and $\mathcal{E} \subseteq \mathcal{V} \times \mathcal{V}$ denotes the set of directed communication links. An edge $(j,i) \in \mathcal{E}$ indicates that agent $i$ can receive information from agent $j$.

Let $\mathcal{K} = \{1,\dots,K\}$ denote the set of arms available in the stochastic environment. We consider the case where each agent can only access arms (select which to pull)  from a nonempty subset 
$\mathcal{K}_i \subseteq \mathcal{K}$. We assume that all arms are collectively accessible, \ie $\cup_{i=1}^N \mathcal{K}_i = \mathcal{K}$.

When agent $i \in \mathcal{V}$ selects arm $k \in \mathcal{K}_i$ at time $t$, it observes a random reward $r_i(t):=X_i^k(t)$. For each arm $k$, the reward distributions are identical across agents and stationary over time, with mean $\mu^k := \mathbb{E}[X_i^k(t)]$, for all $i \in \mathcal{V}$ and for all $t\in\mathbb{N}$. Rewards are assumed to be mutually independent across agents, arms, and time, and have bounded support in $[0,1]$. At each time $t=1,2,\dots,T$, each agent selects exactly one arm $a_i(t) \in \mathcal{K}_i$ and observes the corresponding reward $X_i^{a_i(t)}(t)$. Note that, we allow multiple agents to select the same arm simultaneously without assuming any interference. Furthermore, we define the globally optimal mean by
\begin{align}\label{eq:globally_optimal_mean}
\mu^\star := \max_{k\in\mathcal{K}} \mu^k,
\end{align}
and let $k^\star \in \arg\max_{k\in\mathcal{K}} \mu^k$ denote the globally optimal arm. Since agents may not have access to $k^\star$, we further define the best accessible mean for agent $i$ as
\begin{align}\label{eq:locally_optimal_mean}
\mu_i^\star := \max_{k\in\mathcal{K}_i} \mu^k.
\end{align}

In this setting, the globally optimal arm may not be accessible to all agents $i\in\mathcal{V}$. Hence, the objective of each agent is to minimize the cumulative pseudo-regret with respect to its best accessible arm. For agent $i$, this is defined as
\begin{align}\label{eq:local_regret}
R_i(T) = \sum_{t=1}^T \left( \mu_i^\star - \mu^{a_i(t)} \right),
\end{align}
where $a_i(t) \in \mathcal{K}_i$ is the arm selected at time $t$. The total network feasible regret is
\begin{align}\label{eq:network_feasible_regret}
R(T) = \sum_{i=1}^N R_i(T)
= \sum_{i=1}^N \sum_{t=1}^T \left( \mu_i^\star - \mu^{a_i(t)} \right).
\end{align}

\begin{remark}
Since some agents may not have access to the globally optimal arm with mean
$\mu^\star := \max_{k\in\mathcal{K}} \mu^k$, the cumulative regret with respect to $\mu^\star$ is generally unavoidable and grows at least linearly in $T$ under heterogeneous arm-access constraints. Indeed, one obtains the decomposition
\begin{align*}
\sum_{i=1}^N \sum_{t=1}^T \left( \mu^\star - \mu^{a_i(t)} \right)
=
\underbrace{\sum_{i=1}^N T(\mu^\star - \mu_i^\star)}_{\text{constraint-induced loss}}
+
\underbrace{R(T)}_{\text{learnable regret}},
\end{align*}
where the first term captures the unavoidable performance loss caused solely by arm accessibility constraints, while $R(T)$ represents the regret due to learning under these constraints.
\end{remark}

Accordingly, we focus on minimizing the expected network constrained regret defined in \eqref{eq:network_feasible_regret}, which constitutes the appropriate performance metric for cooperative multi-armed bandits with arm accessibility constraints. Achieving sublinear regret, implies that each agent learns to identify and exploit its best accessible arm over time. Cooperation through information exchange is essential, as no single agent may be able to explore all arms independently.

\subsection{Cooperative MAB under Arm-Accessibility}

In cooperative multi-agent multi-armed bandits (MAMABs), learning performance depends not only on the reward structure, but also on which agents are able to generate information for each arm. Under heterogeneous arm-access constraints, this information generation pattern becomes a structural property of the networked learning system. In this subsection, we formalize this structure and characterize how it interacts with the directed communication topology. An example of the proposed framework is depicted in Fig.\!~\ref{fig:framework}.

The information generation structure captures which agents can generate observations for each arm, and it plays a crucial role in cooperative multi-armed bandits, as it directly affects the effective learning speed. We formalize this structure via the \emph{arm-access matrix}, defined as the binary matrix $C = [C_{ik}] \in \{0,1\}^{N \times K}$,
with entries
\begin{align}
C_{ik} =
\begin{cases}
1, & \text{if agent } i \text{ can pull arm } k,\\
0, & \text{otherwise}.
\end{cases}
\end{align}
Each row of $C$ corresponds to an agent, and each column corresponds to an arm. The support of the $k$-th column, \ie $\mathcal{V}^k := \{ i \in \mathcal{V} : C_{ik} = 1 \}$, represents the set of agents that can generate samples for arm $k$. For each arm $k\in \mathcal{K}$, we define the \emph{generation mass} as
\begin{align}
	g^k:=|\mathcal{V}^k|\equiv\sum_{i\in\mathcal{V}} C_{ik}, \quad \forall\ k\in\mathcal{K}.
\end{align}  
This quantity measures how many agents are able to generate information for arm $k$. Therefore, for the cooperative MAB problem to be well-defined, the arm-access matrix must satisfy the following property.
\begin{assumption}\label{ass:1}
Each arm is accessible by at least one agent, \ie $\sum_{i\in\mathcal{V}} C_{ik} \ge 1, \forall\ k \in \mathcal{K}$. 
\end{assumption}
This assumption guarantees that all arms can potentially be sampled and learned through the network. Without this condition, arms that are inaccessible to all agents would remain permanently unobserved, making the identification of the globally optimal arm impossible. Here it is important to note that, in the distributed setting considered here, agents are not aware of the arm-access capabilities of their neighbors. Hence, the generation mass $g^k$ is unknown to all agents and must be accounted for via information exchange and consensus.

Before introducing mechanisms to address this interaction between network influence and arm-accessibility constraints, we state the following assumptions.

\begin{assumption}\label{ass:2}
All agents know the size of the network $N=|\mathcal{V}|$ and the total number of arms $K=|\mathcal{K}|$.
\end{assumption}

\begin{assumption}\label{ass:3}
The communication digraph $\mathcal{G}=(\mathcal{V},\mathcal{E})$ is strongly connected. 
\end{assumption}

Assumption~\ref{ass:2} enables the design of consensus-based mechanisms that propagate information about all arms, including those not locally accessible to a given agent. If agents do not know the network size $N$ \emph{a~priori}, they can employ distributed network size estimation algorithms, which converge asymptotically \cite{shames2012distributed}, or in finite-time \cite{rikos2023distributed}. Assumption~\ref{ass:3} ensures that information generated by any agent is preserved and can propagate throughout the network.

\begin{assumption}\label{ass:4}
For each agent $i\in\mathcal{V}$, the observed rewards $r_i(t)$ are bounded in $[0,1]$, independent across agents and arms, and identically distributed over time with mean $\mu^k$ for each arm $k\in\mathcal{K}$. Each agent’s action at time $t$ is represented by the pull indicator $\gamma_i^k(t) := \mathbf{1}\{ a_i(t) = k \}$, which equals $1$ if agent $i$ selects arm $k \in \mathcal{K}_i$ at time $t$, and $0$ otherwise.
\end{assumption}

Assumption~\ref{ass:4} is standard in the multi-armed bandit literature. Bounded rewards ensure that empirical averages converge to the true means. Independence across agents and arms guarantees that each local observation provides unbiased information for the network, enabling accurate aggregation of rewards in cooperative learning. 
 
\section{Algorithm Development}
In what follows we design a distributed cooperative algorithm, hereinafter referred to as A2C-UCB (Arm Access Constrained Cooperative Upper Confidence Bound), for solving the multi-agent multi-armed bandits problem with arm-pull constraints with directional communication. We start by formally defining some local variables that each agent maintains to estimate the mean reward which will drive its action towards learning its locally optimal arm. Throughout the algorithm description, we adopt the pull indicator $\gamma_i^k(t)$ as defined in Assumption~\ref{ass:4} for notational convenience.

\subsection{Local sampling (pull) counter and cumulative rewards}
For each agent $i \in \mathcal{V}$ and each arm $k \in \mathcal{K}_i$, we define the \emph{local sampling counter} (\ie the number of times agent $i$ has pulled arm $k$) up to time $t \in \mathbb{N}$ as
\begin{align}\label{eq:cumulative_pull}
n_{i}^{k}(t)
= n_{i}^{k}(t-1) + \gamma_{i}^{k}(t), 
\quad t \ge 1,
\end{align}
with initial condition $n_{i}^{k}(0) = 0$. Similarly, the \emph{local cumulative reward} collected by agent $i$ from arm $k$ up to time $t$ is defined as
\begin{align}\label{eq:cumulative_reward}
s_{i}^{k}(t)
= s_{i}^{k}(t-1) + \gamma_{i}^{k}(t)\, r_i(t), 
\quad t \ge 1,
\end{align}
with $s_i^k(0) = 0$, and where $r_i(t) := X_i^{a_i(t)}(t)$ denotes the random reward obtained by agent $i$ at time $t$ after pulling arm $a_i(t)$. These definitions provide the fundamental building blocks for computing local estimates and confidence bounds in distributed UCB algorithms, allowing each agent to track both its own experience, as well as network-wide statistics via consensus.

\subsection{Distributed empirical mean estimation}
\label{sec:empirical_mean_estimation}
Here, we study a cooperative consensus-based estimation method for network-wide statistics. This method is used by the agents to make sequential actions among arms to maximize their own individual rewards. Prior introducing the cooperative consensus-based estimation scheme, it is important to emphasize that, although agents do not necessarily have access to all arms due to arm-pull constraints, they do have the knowledge of their presence and thus they can also estimate the mean rewards for arms they cannot even pull. In particular, each agent $i \in \mathcal{V}$ maintains an estimate for the total global sample count per arm, denoted by $\hat{n}_{i}^{k}(t)$, and the total global reward collected per arm, denoted by $\hat{s}_{i}^{k}(t)$. These \emph{local statistics} are exchanged within the network respecting the direction of the communication channels and the topology of the network. At each iteration $t\ge0$, agent $i$ sends its local estimates $\hat{s}_{i}^{k}(t)$ and $\hat{n}_{i}^{k}(t)$ to its out-neighboring agents $l \in \outneighbor{i}$, using pre-scaled weights. In particular, each agent $i\in\mathcal{V}$ applies a weight on the values to be sent to its out-neighbors according to its out-degree $\outdegree{i}$. These weights are assigned by each $i\in\mathcal{V}$ for all $j\in\outneighbor{i}$ as:
\begin{align}\label{eq:weights}
    p_{ji}=\begin{cases}
    \frac{1}{1+\outdegree{i}}, & \text{if } j \in \outneighbor{i} \cup \{i\}\\
    0, & \text{otherwise}.
    \end{cases}
\end{align}
Note that, the assignment of the column-stochastic weights above require each agent to know its out-degree.

Upon the reception of $p_{ij} \hat{s}_{j}^{k}(t)$ and $p_{ij} \hat{n}_{j}^{k}(t)$ from its in-neighbors $j\in\inneighbor{i}$, it updates its estimates for each arm (accessible or not) $k\in\mathcal{K}$ using two independent running consensus iterations as follows
\begin{subequations}\label{eq:ratio}
\begin{align}
	\hat{s}_{i}^{k}(t+1) &= \sum_{j\in\inneighbor{i}} p_{ij} \hat{s}_{j}^{k}(t) + \gamma_{i}^{k}(t) r_i(t),\label{eq:ratio_s}\\
	\hat{n}_{i}^{k}(t+1) &= \sum_{j\in\inneighbor{i}} p_{ij} \hat{n}_{j}^{k}(t) + \gamma_{i}^{k}(t),\label{eq:ratio_n}
\end{align}
\end{subequations}
where $\hat{s}_{i}^{k}(0)=0$ and $\hat{n}_{i}^{k}(0)=0$ for all $i\in\mathcal{V}$ and all $k\in\mathcal{K}$. Inspired by the ratio consensus algorithm and its variants (see \cite{hadjicostis2018distributed} for a nice overview), each agent $i\in\mathcal{V}$ can take the ratio of the estimates $\hat{s}_{i}^{k}(t)$ over $\hat{n}_{i}^{k}(t)$ to obtain an estimate on the empirical mean of arm $k\in\mathcal{K}$ at time $t$ defined by 
\begin{align}\label{eq:ratio_estimate}
	\hat{\mu}_{i}^{k}(t) &= \frac{\hat{s}_{i}^{k}(t)}{\max(\hat{n}_{i}^{k}(t),1)}.
\end{align}
This running ratio consensus variant ensures that the distributed updates preserve the correct network-wide totals and asymptotically recover the true arm reward means. In particular, it guarantees that each agent’s contribution is naturally weighted according to its number of pulls for each arm. As a result, agents that sample an arm more frequently have a proportionally larger influence on the estimated mean, ensuring consistency with the centralized empirical average. This property establishes the correct mixing behavior required for convergence of the ratio-consensus updates. Consequently, the network-wide sums, \ie 
\begin{subequations}\label{eq:sum_preservation}
\begin{align}
\sum_{i\in\mathcal{V}} \hat{s}_i^k(t) &= \sum_{\tau=1}^t \sum_{i\in\mathcal{V}} \gamma_i^k(\tau) r_i(\tau),\\
\sum_{i\in\mathcal{V}} \hat{n}_i^k(t) &= \sum_{\tau=1}^t \sum_{i\in\mathcal{V}} \gamma_i^k(\tau),
\end{align}
\end{subequations}
exactly match the true cumulative rewards and sample counts for each arm $k$, so that $\hat{\mu}_i^k(t)$ provides an unbiased estimate of the network-wide arm's empirical mean.

Here, it is important to mention that, if some agent $i\in\mathcal{V}$ cannot access an arm $k$, \ie $k\notin\mathcal{K}_i$, then, the agent's mean reward for that arm will be entirely based only on the received information from its in-neighboring agents.

\subsection{Distributed arm pull count and accessibility estimation}\label{sec:pull_count_estimation}

Building on the distributed empirical mean estimates, we design a cooperative variant of the upper confidence bound (UCB) algorithm for agents with arm-access constraints. Each agent selects arms it can access using local empirical means combined with confidence bounds to balance exploration and exploitation. The algorithm relies on the network-wide unbiased estimates from ratio-consensus \eqref{eq:ratio_estimate}. To properly account for exploration in the UCB function, it is necessary to track the effective sample size, that is the total number of pulls contributing to each arm’s mean estimate, across the network in a distributed manner. Specifically, to estimate this, each agent uses its local estimate regarding the total number of pull counts for each arm $\hat{n}_{i}^{k}(t)$ as updated in \eqref{eq:ratio_n}, and an auxiliary local variable $\hat{y}_{i}(t)$, that is updated as:
\begin{align}\label{eq:y_iter}
	\hat{y}_{i}(t+1) &= \sum_{j\in\inneighbor{i}} p_{ij} \hat{y}_{j}(t),
\end{align}
with initial conditions $\hat{y}_{i}(0)=1$. Computing the ratio $\hat{n}_{i}^{k}(t)/\hat{y}_{i}(t)$, agent $i$ implements a running ratio consensus correction to obtain an estimate of the average pull count for arm $k$ over all agents in $\mathcal{V}$. Since both $\hat n_i^k(t)$ and $\hat y_i(t)$ are mixed through the same primitive column-stochastic matrix $P$, the ratio compensates for the stationary distribution bias induced by the imbalanced information flow due to the directed communication. Note that, however, the pull counts introduce persistent innovations at each time step $t$. Therefore, the ratio does not converge to a fixed value, but instead, it performs a dynamic average tracking of the time-varying network average $\frac{1}{N} \sum_{i\in\mathcal{V}} n_i^k(t)$,
where $n_i^k(t)$ denotes the actual local pull count of arm $k$ at agent $i$. Hence, each agent can track the actual network-wide pull count $n^k(t)=\sum_{i\in\mathcal{V}} n_i^k(t)$ for each arm $k\in\mathcal{K}$, by computing $\hat{\nu}_{i}^{k}(t)=N \hat{n}_{i}^{k}(t) / \hat{y}_{i}(t)$.

To properly normalize exploration effort under heterogeneous accessibility, each agent must estimate, for every arm $k$, the number of agents capable of sampling it. Recall that the number of agents that can access arm $k$ is given by the generation mass of the arm, defined as $g^k := \sum_{i\in\mathcal{V}} C_{ik}$, where $C_{ik}\in\{0,1\}$ is the agent–arm incidence matrix. Since accessibility information is local, no agent directly knows $g^k$. However, unlike pull counts, the accessibility structure is static. Therefore, the problem reduces to distributed computation of a fixed network-wide sum. To this end, each agent $i\in\mathcal V$ maintains a local variable $\hat{u}_{i}^{k}(t)$ which is updated as
\begin{align}\label{eq:generation_mass_refined}
\hat u_i^k(t+1) &= \sum_{j\in\inneighbor{i}} p_{ij}\hat u_j^k(t),
\quad \forall k\in\mathcal K,
\end{align}
with $\hat u_i^k(0)=C_{ik}\in\{0,1\}$, along with the iteration in \eqref{eq:y_iter}. With these values at time $t$, each agent can take the ratio $\hat u_i^k(t)/\hat{y}_i(t)$ to obtain the average number of agents that can pull arm $k$. Since both sequences evolve under the same primitive column-stochastic matrix $P$, their ratio eliminates the stationary-distribution bias induced by directed communication. Consequently, under Assumption~\ref{ass:2} and from \cite{dominguez2011distributed}, the ratio $\hat u_i^k(t)/\hat y_i(t)$ converges to $\frac{1}{N} g^k$ as $t\to\infty$, for every agent $i\in\mathcal{V}$ and arm $k\in\mathcal{K}$. Therefore, each agent can asymptotically obtain the exact generation mass of each arm $k$ by computing $\hat g_i^k(t) = N \hat u_i^k(t)/\hat y_i(t)$. Since $g^k \ge 1$ for all arms by Assumption~\ref{ass:2}, the normalization is well-defined. Note that, because the accessibility structure is assumed time-invariant, the above ratio consensus mechanism converges exponentially fast and does not suffer from persistent innovation, in contrast to the pull-count tracking mechanism.

\subsection{Distributed UCB in the Arm-Restricted Setting}

Using the distributed estimates of network-wide pull counts $\hat{\nu}_i^k(t)$ and generation masses $\hat{g}_i^k(t)$, each agent can construct UCB indices that account for both exploration uncertainty and heterogeneous arm access. Recall that in classical single-agent UCB1 \cite{auer2002finite}, the empirical mean $\bar{\mu}^k(t) = s^k(t)/n^k(t)$ satisfies Hoeffding’s inequality
\begin{align}
\mathbb{P}\big[|\bar{\mu}^k(t)-\mu^k| \ge \epsilon \big] \le 2 \exp(-2 n^k(t) \epsilon^2),
\end{align}
with which having $\epsilon(t)=\sqrt{(\alpha \log t)/n^k(t)}$ gives the standard UCB1 index $\bar{\mu}^k(t) + \sqrt{\frac{\alpha \log t}{n^k(t)}}$. In the cooperative, arm-access constrained setting, each agent $i$ replaces the empirical mean and pull count by the network-weighted estimate $\hat{\mu}_i^k(t)$, the pull counts of each arm is replaced by the distributed estimate of pull counts $\hat{\nu}_i^k(t)$, and further, it accounts for the arm's generation mass using the distributed estimate $\hat{g}_i^k(t)$. This yields the A2C-UCB index
\begin{align}\label{eq:ucb_bound}
Q_i^k(t) := \hat{\mu}_i^k(t) + \underbrace{\sqrt{\frac{\alpha \log \big( t \hat{g}_i^k(t)\big)}
{\hat{\nu}_i^k(t)}}}_{B_i^k(t)},
\end{align}
which balances exploration and exploitation under network asymmetry and arm-access constraints.

\begin{remark}
The denominator $\hat{\nu}_i^k(t)$ reflects the effective global sample size accumulated through cooperation, while the factor $\hat{g}_i^k(t)$ enlarges exploration for arms accessible to fewer agents, compensating for slower information propagation. 
\end{remark}

Now we are in place to summarize the proposed A2C-UCB algorithm from the perspective of a agent $i \in \mathcal{V}$, as shown in Algorithm~\ref{alg:A2C-UCB}. Each agent maintains local estimates of global pull counts and arm generation mass using running ratio consensus, and uses these estimates to compute a UCB index that accounts for both network asymmetry and heterogeneous arm accessibility.

\begin{algorithm}[t]
\caption{A2C-UCB: Arm Access Constrained Cooperative UCB (Agent $i$ perspective)}
\label{alg:A2C-UCB}  
\KwIn{
Number of arms $K$, size of the network $N$,\;
horizon $T$, accessible arms $\mathcal{K}_i$, $\alpha > 1$
}
\KwInit{
$n_i^k(0) = 0$, $\forall k \in \mathcal{K}_i$;\;
$\hat{s}_i^k(0) = 0$, $\hat{n}_i^k(0) = 0$, $\hat{\mu}_i^k(0) = 0$ $\forall k \in \mathcal{K}$;\;
$\hat{u}_i^k(0) = 1$ if $k \in \mathcal{K}_i$, else $0$;\;
$\hat{y}_i(0) = 1$;\;
}
\For{$t = 1$ \KwTo $T$}{
    \ForEach{$k \in \mathcal{K}_i$}{
        $Q_i^k(t) \gets \hat{\mu}_i^k(t) + \sqrt{\frac{\alpha \log \big(t \hat{g}_{i}^{k}(t)\big)}{\hat{\nu}_{i}^{k}(t)}}$\;
    }
    $a_i(t) \gets \arg\max_{k \in \mathcal{K}_i} Q_i^k(t)$\;
    Obtain reward $r_i(t)$ from arm $a_i(t)$\;
    
    Update local variables\;
    \Indp
    $\gamma_{i}^{k}(t) = \mathbf{1}\{a_i(t) = k\}$\;
    $n_{i}^{k}(t) = n_{i}^{k}(t{-}1) + \gamma_{i}^{k}(t)$\;
    \Indm

        Broadcast $p_{\ell i} \hat{y}_i(t)$ to all $\ell \in \outneighbor{i}$\;
        Receive $p_{ij}\hat{y}_j(t)$ from all $j \in \inneighbor{i}$\;
        Update $\hat{y}_i(t+1) = \sum_{j \in \inneighbor{i}} p_{ij} \, \hat{y}_j(t)$\;

    \ForEach{$k \in \mathcal{K}$}{
        Broadcast $p_{\ell i}\big( \hat{s}_i^k(t), \hat{n}_i^k(t), \hat{u}_i^k(t), \hat{y}_i(t) \big)$ to all $\ell \in \outneighbor{i}$\;
        Receive $p_{ij}\hat{s}_j^k(t), p_{ij}\hat{n}_j^k(t), p_{ij}\hat{u}_j^k(t), p_{ij}\hat{y}_j(t)$ from all $j \in \inneighbor{i}$\;
        
        Update local estimates\;%
        $\hat{s}_i^k(t+1) = \sum_{j \in \inneighbor{i}} p_{ij} \, \hat{s}_j^k(t) + \gamma_i^k(t) r_i(t)$\;
        $\hat{n}_i^k(t+1) = \sum_{j \in \inneighbor{i}} p_{ij} \, \hat{n}_j^k(t) + \gamma_i^k(t)$\;
        $\hat{u}_i^k(t+1) = \sum_{j \in \inneighbor{i}} p_{ij} \, \hat{u}_j^k(t)$\;

        $\hat{\mu}_i^k(t) = \hat{s}_i^k(t) / \max(\hat{n}_i^k(t),1)$\;
        $\hat{g}_i^k(t) = N \hat{u}_i^k(t) / \hat{y}_i(t)$\;
        $\hat{\nu}_i^k(t) = N \hat{n}_i^k(t) / \hat{y}_i(t)$\;   
    }
}
\KwOut{Estimated means $\hat{\mu}_i^k(t)$ and selected arm $a_i(t)$
}
\end{algorithm}

\begin{remark}
A fundamental property of Algorithm~\ref{alg:A2C-UCB} is that the estimates $\hat{s}_i^k(t)$ and $\hat{n}_i^k(t)$ track the exact network-wide cumulative reward and number of pulls of arm $k$. This follows from the sum-preserving running ratio consensus algorithm with column-stochastic weights in \eqref{eq:ratio}. Consequently, despite the continuous injection of innovations $\gamma_i^k(t) r_i(t)$ and $\gamma_i^k(t)$, the network maintains the exact cumulative statistics of each arm, and the empirical means $\hat{\mu}_i^k(t)$ provide unbiased estimates of the true network-wide mean reward. In contrast, several distributed bandit algorithms relying on row-stochastic mixing \cite{landgren2016adistributed,landgren2016bdistributed,martinez2019decentralized,landgren2021distributed} or consensus updates that do not preserve the global sums under dynamic innovations \cite{zhu2021distributed,zhu2023distributed,zhu2024decentralized} generally introduce bias in the reward estimates.
\end{remark}

\section{Analysis}

In this section, we analyze the tracking properties of the cooperative estimation mechanism employed by each agent for all $k\in\mathcal{K}$. Recall that, each agent $i$ maintains an estimate $\hat{s}_i^k(t)$ for the cumulative reward up to time $t$ provided by arm $k$, and an estimate $\hat{n}_i^k(t)$ for the total number of pulls of arm $k$ across all agents that can access it. Stacking the agents' estimates into the vectors:
\begin{subequations}
\begin{align}
\hat{s}^k(t) &:= [\hat{s}_1^k(t), \ldots, \hat{s}_N^k(t)]^\top,\\
\hat{n}^k(t) &:= [\hat{n}_1^k(t), \ldots, \hat{n}_N^k(t)]^\top,
\end{align}
\end{subequations}
the running consensus updates can be expressed in compact matrix form as
\begin{subequations}\label{eq:running_consensus}
\begin{align}
\hat{s}^k(t+1) &= P \hat{s}^k(t) + r^k(t), \label{eq:running_consensus_a}\\
\hat{n}^k(t+1) &= P \hat{n}^k(t) + \gamma^k(t),
\end{align}
\end{subequations}
where $P\in\mathbb{R}^{N\times N}$ is the column-stochastic consensus matrix, $r^k(t):=\begin{bmatrix} r_1^{k}(t), \ldots, r_N^{k}(t) \end{bmatrix}^\top$ is the vector of instantaneous rewards, and $\gamma^k(t):=\begin{bmatrix} \gamma_1^{k}(t), \ldots, \gamma_N^{k}(t) \end{bmatrix}^\top$ is the pull-indicator vector for arm $k$.

\subsection{Running Ratio Consensus Tracking}

To analyze the cooperative estimation mechanism, we first study the running consensus dynamics. Intuitively, each agent updates its cumulative reward estimate using local observations and information received from its neighbors. Lemma~\ref{lem:dynamic_tracking} establishes that, even in a directed and unbalanced network, the collective estimate contracts toward a Perron-weighted aggregate in the absence of rewards and remains uniformly bounded when the rewards are bounded. Let $\phi \in \mathbb{R}_{>0}^N$ denote the Perron-Frobenius eigenvector of the column-stochastic primitive matrix $P$, satisfying $P\phi=\phi$ and $\mathbf{1}^\top \phi = 1$ \cite{horn2012matrix}.

\begin{lemma}[Dynamic tracking of running ratio consensus]\label{lem:dynamic_tracking}
Consider the recursion of the estimates of the cumulative rewards in \eqref{eq:running_consensus_a} for some arm $k\in\mathcal{K}$, with $\hat{s}(0)=\mathbf{0}_N$. Suppose the communication graph $\mathcal{G}=(\mathcal{V},\mathcal{E})$ satisfies Assumption~\ref{ass:3}, associated with the weight matrix $P$, and the rewards satisfy Assumption~\ref{ass:4}. Let $e(t) := \hat{s}(t) - \phi \mathbf{1}^\top s(t)$, where $s(t) := \sum_{\tau=0}^{t-1} r(\tau)$, with $r(\tau):=\begin{bmatrix} r_1(\tau), \ldots, r_N(\tau) \end{bmatrix}^\top$. Then, the disagreement $e(t)$ is uniformly bounded for all $t \ge 0$. Furthermore, in the absence of innovations, \ie if there exists $t_0$ such that $r(t)=0$ for all $t \ge t_0$, the disagreement decays exponentially to zero.
\end{lemma}

\begin{proof}
Consider the evolution of the disagreement, \ie
\begin{align}\label{eq:error_dynamics_corrected}
e(t+1)
&= \hat{s}(t+1) - \phi\mathbf{1}^\top s(t+1) \nonumber\\
&= P\hat{s}(t) + r(t) - \phi\mathbf{1}^\top\big(s(t)+r(t)\big) \nonumber\\
&= P e(t) + \big(I-\phi\mathbf{1}^\top\big) r(t).
\end{align}
Note that $\phi$ is the right Perron-Frobenius  eigenvector of $P$ satisfying $P\phi=\phi$ and normalized such that $\mathbf{1}^\top\phi=1$. Thus, $\mathbf{1}^\top e(t)=0$ for all $t$, and therefore the error vector always lies in the disagreement subspace $\{x \in \mathbb{R}^N : \mathbf{1}^\top x = 0\}$. Moreover, since $P$ is primitive and column-stochastic, its spectral radius on this subspace is strictly smaller than one. Consequently, there exist constants $\sigma>0$ and $\rho\in(0,1)$ such that $\|P^t x\|_\infty \le \sigma \rho^t \|x\|_\infty$ for any vector $x$ satisfying $\mathbf{1}^\top x=0$. Now, iterating \eqref{eq:error_dynamics_corrected} gives
\begin{align}
e(t) &= P^t e(0)
+ \sum_{\tau=0}^{t-1} P^{t-1-\tau} (I-\phi\mathbf{1}^\top) r(\tau).
\end{align}
Taking the infinity norm and using submultiplicativity yields
\begin{align}
\|e(t)\|_\infty &\le \|P^t e(0)\|_\infty \nonumber\\
&+ \sum_{\tau=0}^{t-1}
\|P^{t-1-\tau}\|_\infty
\|(I-\phi\mathbf{1}^\top) r(\tau)\|_\infty.
\end{align}
Recall that from Assumption~\ref{ass:4} we know that $0\le r_i(\tau)\le 1$, thus $\|(I-\phi\mathbf{1}^\top) r(\tau)\|_\infty
\le \|I-\phi\mathbf{1}^\top\|_\infty$. Since $(I-\phi\mathbf{1}^\top) r(\tau)$ lies in the disagreement subspace, the above contraction property of matrix $P$ implies
\begin{align}
\|e(t)\|_\infty
&\le \sigma\rho^t\|e(0)\|_\infty + \sigma\|I-\phi\mathbf{1}^\top\|_\infty
\sum_{\tau=0}^{t-1}\rho^\tau \nonumber\\
&= \sigma\rho^t\|e(0)\|_\infty + \sigma\|I-\phi\mathbf{1}^\top\|_\infty \frac{1-\rho^t}{1-\rho},
\end{align}
which shows that $e(t)$ is uniformly bounded. Further, defining
\begin{align}
c_P :=
\sigma\|e(0)\|_\infty
+
\frac{\sigma\|I-\phi\mathbf{1}^\top\|_\infty}{1-\rho},
\end{align}
we obtain $\|e(t)\|_\infty \le c_P$ for all $t\ge0$. The same proof applies for the pull-count estimates $\hat{n}_i^k(t)$.
\end{proof}

\subsection{Network Empirical Mean and Local Tracking Error}

Next, we quantify how well each agent’s local estimate tracks the network empirical mean. Since agents contribute unequally to the network, each local estimate may differ from the true network empirical mean which is defined by $\bar{\mu}^k(t) := \bar{s}^k(t)/\bar{n}^k(t)$, where
\begin{align}\label{eq:weighted_true_vars}
\bar{s}^k(t) := \sum_{i\in\mathcal{V}} \phi_i s_i^k(t), \quad
\bar{n}^k(t) := \sum_{i\in\mathcal{V}} \phi_i n_i^k(t).
\end{align}

\begin{lemma}[Consensus tracking error]\label{lem:tracking_error}
Under Assumption~\ref{ass:3}, let $\hat{s}_i^k(t)$ and $\hat{n}_i^k(t)$ be agent $i$'s estimates for arm $k$. Then, there exists a constant $c_P>0$, depending only on the network connectivity and topology, such that
\begin{align}
\big| \hat{\mu}_i^k(t) - \bar{\mu}^k(t) \big| 
\le \frac{2 c_P}{\bar{n}^k(t) - c_P},
\end{align}
where the network-weighted empirical mean is defined by $\bar{\mu}^k(t) := \bar{s}^k(t)/\bar{n}^k(t)$ with $\bar{s}^k(t) := \sum_{i\in\mathcal{V}} \phi_i s_i^k(t)$, $\bar{n}^k(t) := \sum_{i\in\mathcal{V}} \phi_i n_i^k(t)$.
\end{lemma}

\begin{proof}
We begin by defining the local errors for agent $i$ as $\delta_i^k(t) := \hat{s}_i^k(t) - \bar{s}^k(t)$ and $\breve{\delta}_i^k(t) := \hat{n}_i^k(t) - \bar{n}^k(t)$. From Lemma~\ref{lem:dynamic_tracking}, the running consensus dynamics guarantee that $|\delta_i^k(t)| \le c_P$ and $|\breve{\delta}_n^k(t)| \le c_P$, for all $t$, where $c_P>0$ depends only on the network topology and the weight matrix $P$. Now consider the difference between the local empirical mean and the network-weighted mean:
\begin{align}
\big| \hat{\mu}_i^k(t) - \bar{\mu}^k(t) \big|
&= \left| \frac{\hat{s}_i^k(t)}{\hat{n}_i^k(t)} - \frac{\bar{s}^k(t)}{\bar{n}^k(t)} \right| \nonumber\\
&= \left| \frac{\delta_i^k(t) \bar{n}^k(t) - \bar{s}^k(t) \breve{\delta}_i^k(t)}{\bar{n}^k(t)(\bar{n}^k(t)+\breve{\delta}_i^k(t))} \right|.
\end{align}
Under Assumption~\ref{ass:4} we have $0 \le \bar{s}^k(t) \le \bar{n}^k(t)$ since rewards are bounded in $[0,1]$. Thus, using $|\delta_i^k(t)| \le c_P$ and $|\breve{\delta}_i^k(t)| \le c_P$, we obtain
\begin{align}
\big| \hat{\mu}_i^k(t) - \bar{\mu}^k(t) \big|
&\le \frac{|\delta_i^k(t)| \bar{n}^k(t) + |\bar{s}^k(t)| |\breve{\delta}_i^k(t)|}{\bar{n}^k(t)(\bar{n}^k(t)-c_P)} \nonumber\\
&\le \frac{2 c_P \bar{n}^k(t)}{\bar{n}^k(t)(\bar{n}^k(t)-c_P)} 
= \frac{2 c_P}{\bar{n}^k(t)-c_P},
\end{align}
where we have used that $\bar{n}^k(t) + \breve{\delta}_i^k(t) \ge \bar{n}^k(t) - c_P > 0$. This completes the proof.
\end{proof}

\begin{remark}
Lemma~\ref{lem:tracking_error} shows that the error due to distributed consensus decreases as more samples are collected. In particular, the network influence, captured by $c_P$, becomes negligible over time, and thus the mean reward estimate of each agent approaches the centralized empirical mean.
\end{remark}

\subsection{Derivation of the A2C-UCB Confidence Radius}

Combining the consensus tracking error with classical concentration inequalities allows us to derive high-probability confidence intervals for each agent’s estimate.

\begin{lemma}[A2C-UCB confidence bound]\label{lem:ucb_bound}
Let $\hat{\mu}_i^k(t)$ be the mean reward estimate of agent $i$ for arm $k$ at time $t$. Then, for any $\alpha>1$ and horizon $T$, with probability at least $1-2 T^{1-\alpha}$ (uniformly for all $0 \le t \le T$),
\begin{align}\label{eq:confidence_bound}
|\hat{\mu}_i^k(t) - \mu^k| \le B_i^k(t) + \frac{2 c_P}{\bar{n}^k(t) - c_P},
\end{align}
where $B_i^k(t) = \sqrt{\frac{\alpha \log (t \hat{g}_i^k(t))}{\hat{\nu}_i^k(t)}}$.
\end{lemma}

\begin{proof}
Recall that, by Lemma~\ref{lem:tracking_error}, the local estimate satisfies
\begin{align}
|\hat{\mu}_i^k(t) - \bar{\mu}^k(t)| \le \frac{2 c_P}{\bar{n}^k(t) - c_P}.
\end{align}
The network-weighted mean $\bar{\mu}^k(t)$ is a convex combination of bounded i.i.d. rewards, so applying Hoeffding’s inequality gives, for any $\epsilon > 0$,
\begin{align}
\mathbb{P}\Big[ |\bar{\mu}^k(t) - \mu^k| \ge \epsilon \Big] \le 2 \exp\big(- 2 \bar{n}^k(t) \epsilon^2 \big).
\end{align}
Now, we choose $\epsilon$ to guarantee a failure probability at time $t$ of at most $2 (t \hat g_i^k(t))^{-\alpha}$. Solving
\begin{align}
2 \exp(-2 \bar{n}^k(t) \epsilon^2) = 2 (t \, \hat g_i^k(t))^{-\alpha}
\end{align}
yields
\begin{align}
\epsilon = \sqrt{\frac{\alpha \log (t \, \hat g_i^k(t))}{2 \bar{n}^k(t)}}.
\end{align}
Replacing $\bar{n}^k(t)$ with the local estimate $\hat{\nu}_i^k(t)$, which tracks it up to a bounded consensus error (Lemma~\ref{lem:tracking_error}), gives the statistical exploration bonus
\begin{align}
B_i^k(t) := \sqrt{\frac{\alpha \log (t \, \hat g_i^k(t))}{\hat{\nu}_i^k(t)}}.
\end{align}
Finally, by the triangle inequality we obtain
\begin{align}
\big| \hat{\mu}_i^k(t) - \mu^k \big| 
&\le \big| \hat{\mu}_i^k(t) - \bar{\mu}^k(t) \big| 
+ \big| \bar{\mu}^k(t) - \mu^k \big| \nonumber\\
&\le \frac{2 c_P}{\bar{n}^k(t) - c_P} + B_i^k(t),
\end{align}
which holds with probability at least $1 - 2 (t \, \hat g_i^k(t))^{-\alpha}$ for each $t \le T$. Applying a union bound over $t=1,\dots,T$, we obtain
\begin{align}
|\hat{\mu}_i^k(t) - \mu^k| \le B_i^k(t) + \frac{2 c_P}{\bar{n}^k(t) - c_P}
\end{align}
with probability at least $1 - 2 T^{1-\alpha}$, for all $t \le T$.
\end{proof}

Lemma~\ref{lem:ucb_bound} follows the classical UCB analysis used in the stochastic MAB literature, which combines concentration inequalities with a high-probability bound on the deviation of empirical means from true rewards \cite{auer2002finite}. In the centralized setting, Hoeffding’s inequality is used to show that, for any arm, the sample mean stays within an $O\!\big(\sqrt{(\log t)/n^k(t)}\big)$ neighborhood of the true mean with high probability. Our result extends this approach to the distributed cooperative case with arm accessibility constraints by incorporating an additional tracking error due to consensus estimation (Lemma~\ref{lem:tracking_error}), yielding confidence intervals that hold uniformly over time. Consequently, each agent’s estimate $\hat{\mu}_i^k(t)$ remains close to the true mean $\mu^k$ up to an exploration radius $B_i^k(t)$ and a bounded network estimation error.


\subsection{Regret analysis of A2C-UCB}

For the subsequent results, we further define for each agent $i\in\mathcal{V}$ and arm $k\in\mathcal{K}_i$, the suboptimality gap by $\Delta_i^k := \mu^{k_i^\star} - \mu^k$, where $k_i^\star := \arg\max_{k\in\mathcal{K}_i} \mu^k$ is the locally optimal arm for agent $i$.

\begin{lemma}[Upper bound on suboptimal pulls]
\label{lem:suboptimal_pulls_correct}
Fix agent $i$ and a suboptimal arm
$k \in \mathcal{K}_i \setminus \{k_i^\star\}$.
On the event where the confidence bounds hold for all $t \le T$, the number of pulls of arm $k$ satisfies
\begin{align}\label{eq:suboptimal_pulls_bound_a2c}
N_i^k(T)
\le
\frac{16 \, \alpha \, \log \big(T \hat g_i^k(T)\big)}{(\Delta_i^k)^2}
+
\frac{4 c_P}{\Delta_i^k}
+
1.
\end{align}
\end{lemma}

\begin{proof}
From the definition of arm selection in A2C-UCB, if arm $k$ is selected at time $t$, then
\begin{align}
\hat{\mu}_i^k(t) + B_i^k(t)
\ge
\hat{\mu}_i^{k_i^\star}(t) + B_i^{k_i^\star}(t),
\end{align}
Next, using the confidence bounds for both arms from Lemma~\ref{lem:ucb_bound}, we have
\begin{align}
\Delta_i^k 
= \mu^{k_i^\star} - \mu^k &\le B_i^k(t) + B_i^{k_i^\star}(t) \nonumber\\ &+ \frac{2 c_P}{\hat \nu_i^k(t)-c_P} + \frac{2 c_P}{\hat \nu_i^{k_i^\star}(t)-c_P}.
\end{align}
Since the network-weighted pull count of the optimal arm is at least that of any suboptimal arm, $\hat \nu_i^{k_i^\star}(t) \ge \hat \nu_i^k(t)$, we can absorb its consensus error term:
\begin{align}
\frac{2 c_P}{\hat \nu_i^k(t)-c_P} + \frac{2 c_P}{\hat \nu_i^{k_i^\star}(t)-c_P} 
\le 
\frac{4 c_P}{\hat \nu_i^k(t)-c_P}.
\end{align}
Now, for sufficiently large $\hat \nu_i^k(t)$ such that the statistical bonus dominates the consensus term, a sufficient condition to stop pulling arm $k$ is $ 2 B_i^k(t) \le \frac{\Delta_i^k}{2}$, which, after substituting $B_i^k(t)$ from \eqref{eq:ucb_bound}, yields
\begin{align}
\hat \nu_i^k(t) \ge \frac{16 \, \alpha \, \log (t \, \hat g_i^k(t))}{(\Delta_i^k)^2}.
\end{align}
Finally, including the consensus error additive term $4 c_P / \Delta_i^k$ gives the stated bound in \eqref{eq:suboptimal_pulls_bound_a2c}, where we also add the $+1$ term to account for the initial pull of the arm, ensuring that the bound holds for all suboptimal arms.
\end{proof}

Finally, we analyze the cumulative regret of A2C-UCB. The regret decomposes into the classical logarithmic UCB term and an additive network-induced error term.

\begin{theorem}
\label{thm:expected_regret}
Consider each agent $i\in\mathcal{V}$ in a directed network $\mathcal{G}=(\mathcal{V},\mathcal{E})$ executing the distributed A2C-UCB algorithm, over a MAMAB framework with arm-accessibility constraints, for a horizon $T$. Then, the individual expected regret of each agent $i$ is bounded as
\begin{align*}
\mathbb{E}[R_i(T)]
&\le
\sum_{k \in \mathcal{K}_i \setminus \{k_i^\star\}}
\left(
\frac{16 \alpha \log \big(T \hat g_i^k(T)\big)}{\Delta_i^k} + 4 c_P + \Delta_i^k
\right) \nonumber\\&+ 2 T^{1-\alpha},
\end{align*}
and for any $\alpha>1$, the regret grows sublinearly as $O(\log T)$.
\end{theorem}

\begin{proof}
Let $\mathcal{S}_T$ denote the event that the confidence bounds hold for all arms and all times $t \le T$. By a union bound over arms and time, and using Lemma~\ref{lem:ucb_bound}, we have
\begin{align}
\mathbb{P}(\mathcal{S}_T^c) \le 2 T^{1-\alpha}.
\end{align}
On the event $\mathcal{S}_T$, Lemma~\ref{lem:suboptimal_pulls_correct} implies that for each suboptimal arm $k \neq k_i^\star$,
\begin{align}
N_i^k(T) \le
\frac{16 \alpha \log \big(T \hat g_i^k(T)\big)}{(\Delta_i^k)^2} + \frac{4 c_P}{\Delta_i^k} + 1.
\end{align}
Multiplying by $\Delta_i^k$ and summing over $k \neq k_i^\star$ gives
\begin{align}
R_i(T) &= \sum_{k \ne k_i^\star} \Delta_i^k N_i^k(T) \nonumber\\
&\le \sum_{k \ne k_i^\star} \left( \frac{16 \alpha \log \big(T \hat g_i^k(T)\big)}{\Delta_i^k} + 4 c_P + \Delta_i^k \right).
\end{align}
Now, on the complement event $\mathcal{S}_T^c$, the regret can be trivially bounded by $R_i(T) \le T$. Therefore, taking expectations and decomposing over the two events yields
\begin{align*}
\mathbb{E}[R_i(T)] 
&\le \sum_{k \ne k_i^\star} \left( \frac{16 \alpha \log \big(T \hat g_i^k(T)\big)}{\Delta_i^k} + 4 c_P + \Delta_i^k \right) \nonumber\\ &+ T \mathbb{P}(\mathcal{S}_T^c),%
\end{align*}
which gives the result. This completes the proof.
\end{proof}

\begin{remark}
The regret bound reflects heterogeneous arm accessibility in the sense that, arms reachable by fewer or low-influence agents take longer to explore. The algorithm relies solely on local communication, and cooperation through consensus ensures that information propagates despite network asymmetries.
\end{remark}

\section{Simulation Results}

In this section, we evaluate the performance of the proposed cooperative multi-agent bandit framework under fixed arm-accessibility constraints. The objective is to illustrate the impact of cooperation and information sharing on cumulative regret. 

Consider an edge computing task offloading problem, where a small network of edge devices (\eg sensors, IoT devices) can offload computational tasks to nearby servers dedicated for a specific task. Each task type (\eg video processing, image recognition, data compression, etc) corresponds to an “arm” in a multi-armed bandit setup. In particular, we consider $N=6$ devices (agents) indexed as $\mathcal{V}=\{1,2,\ldots,6\}$, each having limited access (due to connectivity, battery, or computation limits) to some of a total of $K=7$ servers (arms)  indexed as $\mathcal{K}=\{1,2,\ldots,7\}$. Each arm $k \in \mathcal{K}$ is associated with an unknown mean reward $\mu^k$. The true arm means are unknown to the agents, yet, they are given here for ease of presentation:
\begin{equation}\label{eq:arm_means}
\bm{\mu} = [0.9,\;0.8,\;0.6,\;0.5,\;0.3,\;0.2,\;0.1].
\end{equation}
The globally optimal arm is arm $1$ with mean $\mu^1=0.9$. Each agent $i$ can only select arms from a fixed accessibility set $\mathcal{K}_i \subseteq \mathcal{K}$. The arm-access structure of the agents is given by the arm-access matrix $C\in\{0,1\}^{N \times K}$ shown in Fig.\!~\ref{fig:simulation_setup}. This structure induces heterogeneous learning capabilities and partial overlap across agent groups. 

\begin{figure}[h!]
  \centering
  \begin{minipage}{.25\textwidth}\    
	\includegraphics[scale=0.7]{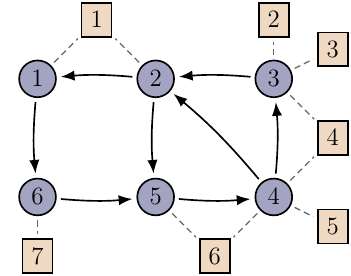}
  \end{minipage}%
  \begin{minipage}{.25\textwidth}
	\vspace{-5pt}
	\begin{align*}
	C =
	\begin{bmatrix}
	1\ 0\ 0\ 0\ 0\ 0\ 0 \\
	1\ 0\ 0\ 0\ 0\ 0\ 0 \\
	0\ 1\ 1\ 1\ 0\ 0\ 0 \\
	0\ 0\ 0\ 1\ 1\ 1\ 0 \\
	0\ 0\ 0\ 0\ 0\ 1\ 0 \\
	0\ 0\ 0\ 0\ 0\ 0\ 1
	\end{bmatrix}
	\end{align*}    
  \end{minipage}
	\caption{The distributed multi-agent multi-armed bandits setup with arm-accessibility constraints. Blue circular nodes represent agents; beige rectangular nodes represent arms. Communication among the agents is shown in black directional arrows, while arm-accessibility per arm (also given in matrix $C$) is depicted with gray dashed lines.}
	\label{fig:simulation_setup}
\end{figure}

At each time $t$, each agent $i$ selects an arm $k \in \mathcal{K}_i$, receives the instantaneous reward $r_{i}(t) = X_i^k(t)$, and updates the pull count of the selected arm. Note that, agents that can only access a single arm cannot explore other arms, but they still contribute to the network-wide estimation of the arms they can access, improving learning for other agents through consensus. In general, agents exchange the local estimates of their cumulative reward and pull counts for each arm in the network (even if the arm is not accessible to them), over the network. For each agent $i$, its optimal accessible arm is denoted by $\mu_i^\star$. In the results that follow, we depict the network cumulative regret, defined as $R(T) = \sum_{i=1}^{N} R_i(T)$, where $R_i(T)$ is the cumulative regret of agent $i$ up to time $T$ as defined in \eqref{eq:local_regret}, and $a_i(t)$ denotes the arm selected at time $t$. All results are averaged over $50$ independent Monte Carlo simulations with the time horizon set to $T=10000$. 

In Fig.\!~\ref{fig:cum_regret}, we compare the proposed cooperative algorithm A2C-UCB with a non-cooperative baseline, where each agent independently runs UCB1 without any information exchange, herein referred to as UCB1 (no comm.). We report the cumulative sum of individual regrets across all agents over time. The proposed cooperative strategy significantly reduces cumulative regret compared to independent learning, since reward observations are shared across the network. This allows agents to update their estimates more efficiently, accelerating the identification of the globally optimal arm and reducing redundant exploration, ultimately leading to substantially lower cumulative regret than the non-cooperative baseline.

To further evaluate the benefits of our proposed algorithm A2C-UCB, we compare it with the cooperative algorithm of Zhu \emph{et al.} \cite{zhu2025decentralized}. Since their method does not account for partial arm-access constraints, we consider a full-arm-access setting in which every agent can sample any arm, using the same mean rewards as in \eqref{eq:arm_means} for a fair comparison. In this experiment, we use a larger directed communication network with $N=30$ agents. As shown in Fig.~\ref{fig:cum_regret_comp}, A2C-UCB achieves significantly lower cumulative regret than the algorithm in \cite{zhu2025decentralized}, demonstrating strong performance even when all arms are accessible. Moreover, in more realistic scenarios with restricted arm access, A2C-UCB naturally accommodates heterogeneous accessibility, whereas the method in \cite{zhu2025decentralized} cannot operate under such constraints.

\begin{figure}[t]
\centering
\includegraphics[scale=0.95]{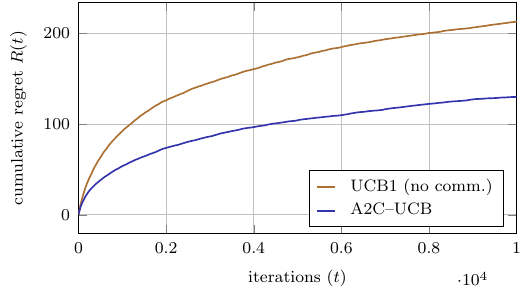}
\caption{Sum of individual agents' cumulative regret for UCB1 (no comm.) and A2C-UCB, under arm-access constraints.}
\label{fig:cum_regret}
\end{figure}

\begin{figure}[t]
\centering
\includegraphics[scale=0.95]{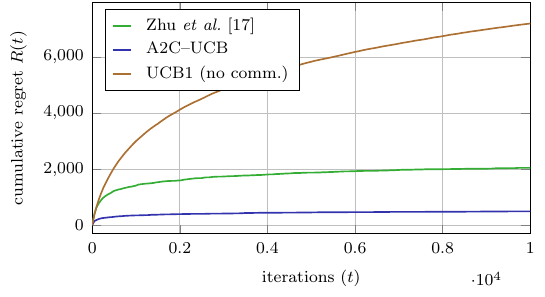}
\caption{Sum of individual agents' cumulative regret for A2C-UCB, UCB1 (no comm.), and \cite{zhu2025decentralized} under full arm accessibility.}
\label{fig:cum_regret_comp}
\end{figure}

\section{Conclusions and Future Work}
In this work, we studied cooperative multi-agent multi-armed bandits with heterogeneous arm access over directed networks. We proposed a distributed consensus-based UCB algorithm that accounts for both accessibility constraints and asymmetric information exchange, ensuring unbiased reward estimation through mass-preserving information mixing. By combining this consensus mechanism with network-aware confidence bounds, we proved logarithmic regret for every agent. Future work includes extending the framework to unreliable or asynchronous communication \cite{hadjicostis2013average,hadjicostis2015robust,makridis2023harnessing}, dynamic networks \cite{makridis2024average}, and time-varying arm availability, that could capture the interplay between accessibility and network topology in a more dynamic environment.


\bibliographystyle{IEEEtran}
\bibliography{references}

\end{document}